\documentclass[12pt]{article}
\usepackage{amssymb}
\usepackage{amsmath}
\usepackage{amsthm}
\usepackage{bm} % used for get math symbols in bold use \bm
\usepackage{graphics}
\usepackage{graphicx}
\usepackage[T1]{fontenc}
\usepackage{longtable}
\usepackage{natbib}
\usepackage{setspace}
\usepackage{times}

\usepackage{epsfig}
\usepackage{fancyvrb}
\usepackage{multirow}
\usepackage[affil-it]{authblk}
\usepackage{verbatim}
\usepackage{tablefootnote}
\usepackage[explicit]{titlesec}
\usepackage{dcolumn}

\usepackage{tikz}
\usetikzlibrary{shapes,snakes, graphs}

\usetikzlibrary{arrows,%
                petri,%
                topaths}%
\usepackage{tkz-berge}
\usepackage{pgf}
\usepackage{subfigure,tikz}
\usetikzlibrary{arrows,automata}

\newtheorem{theorem}{Theorem}
\newtheorem{corollary}{Corollary}
\newtheorem*{theorem*}{Theorem}

\def\PDE{\textsc{pde}}
\def\NIE{\textsc{nie}}
\def\NPSEM{\textsc{npsem-ie}}

\def\DAG{\textsc{dag}}
\def\SWIG{\textsc{swig}}

\def\ci{\mbox{\ensuremath{\perp\!\!\!\perp}}}
\def\pr{\mathrm{pr}}

\makeatletter
\providecommand{\@LN}[2]{}
\makeatother

\titleformat{\section}
  {\normalfont\Large\bfseries\centering}{\thesection.}{1em}{\MakeUppercase{#1}}
  
\begin{document}

\def\spacingset#1{\renewcommand{\baselinestretch}%
{#1}\small\normalsize} \spacingset{1}

\title{\bf On Partial Identification of the Pure Direct Effect}
\author{Caleb Miles, Phyllis Kanki, Seema Meloni, and Eric Tchetgen Tchetgen\thanks{Caleb Miles is Postdoctoral Fellow, Department of Biostatistics, University of California, Berkeley 94720-7358. Phyllis Kanki is Professor and Seema Meloni is Research Associate, Department of Immunology and Infectious Diseases, Harvard School of Public Health, Boston, MA 02115. Eric Tchetgen Tchetgen is Professor, Departments of Biostatistics and Epidemiology, Harvard School of Public Health, Boston, MA 02115. The authors gratefully acknowledge the hard work and dedication of the clinical, data, and laboratory staff at the PEPFAR supported Harvard/AIDS Prevention Initiative in Nigeria (APIN) hospitals that provided secondary data for this analysis. This work was funded, in part, by grants from the U.S. National Institutes of Health. The contents are solely the responsibility of the authors and do not represent the official views of the funding institutions.}\hspace{.2cm}}
\date{}
\maketitle
\bigskip

\begin{abstract}
\noindent In causal mediation analysis, nonparametric identification of the pure (natural) direct effect typically relies on, in addition to no unobserved pre-exposure confounding, fundamental assumptions of (i) so-called ``cross-world-counterfactuals" independence and (ii) no exposure-induced confounding. When the mediator is binary, bounds for partial identification have been given when neither assumption is made, or alternatively when assuming only (ii). We extend existing bounds to the case of a polytomous mediator, and provide bounds for the case assuming only (i). We apply these bounds to data from the Harvard PEPFAR program in Nigeria, where we evaluate the extent to which the effects of antiretroviral therapy on virological failure are mediated by a patient's adherence, and show that inference on this effect is somewhat sensitive to model assumptions.
\end{abstract}

\noindent%
\vfill

\newpage
\spacingset{1.45} % DON'T change the spacing!

%%%%%%%%%%%%%%%%%%%%%%%%%%%%%%%%%%%%%%%%%%%%%%%%%%%%%%%%%%%%%%%%%%%%%%%%%%%%%%%%%%%%%%%%%%%%%%%%%%%%%%%%%
\section{Introduction}
%%%%%%%%%%%%%%%%%%%%%%%%%%%%%%%%%%%%%%%%%%%%%%%%%%%%%%%%%%%%%%%%%%%%%%%%%%%%%%%%%%%%%%%%%%%%%%%%%%%%%%%%%

Causal mediation analysis seeks to determine the role that an intermediate variable plays in transmitting the effect from an exposure to an outcome. An indirect effect refers to the effect that goes through the intermediate variable in mediation analysis; a direct effect is a measure of the effect that does not. The study of causal mediation has in recent years enjoyed an explosion in popularity \citep{robins1992identifiability,robins1999testing,robins2003semantics,pearl2001direct,
avin2005identifiability,taylor2005counterfactual,petersen2006estimation,ten2007causal,
albert2008mediation,goetgeluk2008estimation,van2008direct,vanderweele2009marginal,
vanderweele2009conceptual,vanderweele2010odds,imai2010general,
imai2010identification,albert2011generalized,tchetgen2011causal,vanderweele2011causal,albert2012mediation,tchetgen2012semiparametric,
wang2012estimation,shpitser2013counterfactual,tchetgen2013inverse,tchetgen2014estimation,wang2013estimation,albert2015sensitivity,hsu2015surrogate}, not only in terms of theoretical developments, but also in practice, most notably in the fields of epidemiology and social sciences. This strand of work is based on ideas originating from \cite{robins1992identifiability} and \cite{pearl2001direct} grounded in the language of potential outcomes \citep{splawa1990application, rubin1974estimating, rubin1978bayesian} to give nonparametric definitions of effects involved in mediation analysis, allowing for settings where interactions and nonlinearities may be present.

Consider an intervention which sets the exposure of interest for all persons in the population to one of two possible values, a reference value or an active value. The total effect of such an intervention corresponds to the change of the counterfactual outcome mean if the exposure were set to the active value compared with if it were set to the reference value. \cite{robins1992identifiability} formalized the concept of effect decomposition of the total effect into direct and indirect effects by defining pure direct and indirect effects. \cite{pearl2001direct} relabeled these effects as natural direct and indirect effects. The pure direct effect ($\PDE$) corresponds to the change in the counterfactual outcome mean under an intervention which changes a person's exposure status from the reference value to the active value, while maintaining the person's mediator to the value it would have had under the exposure reference value. In contrast, the natural indirect effect ($\NIE$) corresponds to the change in the average counterfactual outcome under an intervention that sets a person's exposure value to the active value, while changing the value of the mediator from the value it would have had under the reference exposure value, to its value under the active exposure value. The $\PDE$ and $\NIE$ sum to give the total effect.

Identification of these natural effects has been somewhat controversial as it requires assumptions that may be overly restrictive for many applications in the health sciences. First, identification invokes a so-called cross-world-counterfactuals-independence assumption, which by virtue of involving counterfactuals under conflicting interventions on the exposure, can neither be enforced experimentally nor tested empirically \citep{pearl2001direct, robins2010alternative}. Secondly, a necessary assumption for identification rules out the presence of exposure-induced confounding of the mediator's effect on the outcome, even if all confounders are observed. While this assumption is in principle testable provided no unmeasured confounding, more often than not, post-exposure covariates are altogether ignored in routine application, in which case mediation analyses may be invalid. These issues have recently been considered, and some work has been done on partial or point identification under a weaker assumption. Specifically, on the one hand \cite{robins2010alternative} and \cite{tchetgen2014identification} provide conditions for point identification of the pure direct effect when a confounder is directly affected by the exposure. On the other hand, \cite{robins2010alternative} give bounds for the pure direct effect for binary mediator without making the cross-world-counterfactual-independence assumption, but assuming no exposure-induced confounding of the mediator-outcome relation, and \cite{tchetgen2014bounds} extend these bounds to account for exposure-induced confounding. Bounds are commonly employed in causal inference when structural assumptions are not sufficiently strong to give point identification of a causal parameter of interest \citep{robins1989analysis,balke1997probabilistic,zhang2003estimation,kaufman2005improved,cheng2006bounds,cai2008bounds,sjolander2009bounds,taguri2015principal}. We build on this previous work to provide a number of new nonparametric bounds for the pure direct effect allowing for a polytomous mediator when either (i) exposure-induced confounding is present, or (ii) one does not assume that cross-world counterfactuals of the mediating and outcome variables are independent, or (iii) both (i) and (ii) hold.

We apply these bounds to data from the Harvard PEPFAR program in Nigeria, where we evaluate the extent to which the effects of antiretroviral therapy on virological failure are mediated by a patient's adherence. We show that PEPFAR results are sensitive to the choice of assumptions made, consequently, we counsel investigators employing these effects to exercise caution in considering the basis for point identification and to explicitly state the assumptions required for them to be valid. Where assumptions are empirically untestable, they should be argued for on the basis of scientific understanding, and ideally the alternative should be explored by employing partial identification bounds given both here and elsewhere. While some work has been done to develop sensitivity analyses for unmeasured confounding of the mediator \citep{tchetgen2011causal,tchetgen2012semiparametric,vansteelandt2012natural}, sensitivity analyses for ranges of plausible associations between cross-world counterfactuals remain undeveloped. Further development of sensitivity analyses of both forms would be highly beneficial for practical use, and is fertile ground for future work. We hope that the work presented here will inspire deeper consideration and transparency regarding underlying identifying assumptions in the practice of mediation analysis.

%%%%%%%%%%%%%%%%%%%%%%%%%%%%%%%%%%%%%%%%%%%%%%%%%%%%%%%%%%%%%%%%%%%%%%%%%%%%%%%%%%%%%%%%%%%%%%%%%%%%%%%%%
\section{Preliminaries}
%%%%%%%%%%%%%%%%%%%%%%%%%%%%%%%%%%%%%%%%%%%%%%%%%%%%%%%%%%%%%%%%%%%%%%%%%%%%%%%%%%%%%%%%%%%%%%%%%%%%%%%%%
By way of introduction, the directed acyclic graph ($\DAG$) displayed in Fig. \ref{fig:1}.(a) illustrates the simplest possible mediation setting, where $A$ is defined to be the exposure taking either baseline value $a^*$ or comparison value $a$, $M$ is defined to be the (potential) mediator, and $Y$ is defined to be the outcome.
\begin{figure}
\centering
\begin{tabular}{ccc}
\\
\\
\begin{tikzpicture}[->,>=stealth',baseline={(A)},scale=1, line width=1pt]
\tikzstyle{every state}=[draw=none]
\node[shape=circle, draw, inner sep=1mm] (A) at (0,0) {$A$};
\node[shape=circle, draw, inner sep=1mm] (M) at (2,0) {$M$};
\node[shape=circle, draw, inner sep=1mm] (Y) at (4,0) {$Y$};

  \path 
	(A)  edge               (M)  
	(A)  edge  [bend right]  (Y)
	(M)  edge              (Y)
	;
\end{tikzpicture}
& 
&
\begin{tikzpicture}[->,>=stealth',baseline={(A)},scale=1, line width=1pt]
\tikzstyle{every state}=[draw=none]
\node[shape=semicircle, draw, inner sep=1mm, shape border rotate=90, inner sep=1.5mm] (A) at (0,0) {$A$};
\node[shape=semicircle, draw, shape border rotate=270, color=red, inner sep=1.85mm] (a) at (.75,0) {$\tilde{a}$};
\node[shape=semicircle, draw, inner sep=1mm, shape border rotate=90, inner sep=.5mm] (M) at (3,0) {$M(\tilde{a})$};
\node[shape=semicircle, draw, shape border rotate=270, color=red, inner sep=2mm] (m) at (4,0) {$\tilde{m}$};
\node[shape=ellipse, draw, inner sep=1mm] (Y) at (6.25,0) {$Y(\tilde{a},\tilde{m})$};

  \path 
	(a)  edge               (M)  
	(m)  edge               (Y)  
	(a)  edge  [bend right=60]   (Y)  
	;
\end{tikzpicture}
\end{tabular}
\caption{(a) The three-node mediation directed acyclic graph in a setting with no confounding. The nodes represent random variables, and the arrows represent possible causal effects of one random variable on another. (b) The single-world intervention graph in the setting of (a) under the intervention setting $A$ to $\tilde{a}$ and $M$ to $\tilde{m}$. The black nodes represent random variables under this intervention, the red nodes represent the level an intervened random variable takes under this intervention, and the arrows represent possible causal effects of one variable under this intervention on another.}
\label{fig:1}
\end{figure}
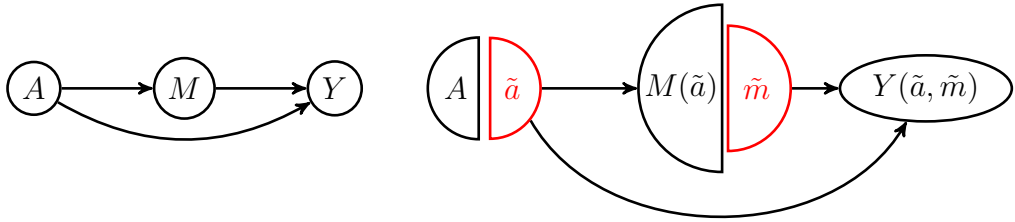
This $\DAG$ assumes randomization of the exposure, which for expositional simplicity we maintain throughout. The graph also encodes no unobserved confounding of the effect of $M$ on $Y$ given $A$. The effect along the path $A\rightarrow Y$ on the diagram is generally referred to as direct with respect to $M$, and the effect along the path $A\rightarrow M\rightarrow Y$ on the diagram is generally referred to as indirect with respect to $M$.

Further elaboration of the specific type of direct and indirect effect under consideration necessitates counterfactual definitions. Let $Y(a)$ denote a subject's outcome if treatment $A$ were set, possibly contrary to fact, to $a$. In the context of mediation, there will also be potential outcomes for the intermediate variable. Counterfactuals $M(a)$ and $Y(m,a)$ are defined similarly. In order to link these with the observed data, we adopt the standard set of consistency assumptions that
\begin{align*}
&\textrm{if } A=a\textrm{, then } M(a)=M\textrm{ with probability one,}\\
&\textrm{if } A=a\textrm{ and } M=m\textrm{, then } Y(m,a)=Y\textrm{ with probability one, and}\\
&\textrm{if } A=a\textrm{, then } Y(a)=Y\textrm{ with probability one.}\\
\end{align*}
In terms of counterfactuals, the randomization assumption encoded by the $\DAG$ in Fig. \ref{fig:1}.(a) is $\{Y(a,m),M(a)\}\ci A$ for all $a$ and $m$; the assumption of no unobserved confounding of $M$ given $A$ is $Y(a,m)\ci M(a)\mid A=a$ for all $a$ and $m$. Finally, we will consider as well defined the nested counterfactual $Y\{a,M(a^*)\}$, i.e., the counterfactual outcome under an intervention which sets the exposure to the comparison value $a$, and the mediator to the value it would have taken under the conflicting baseline exposure value $a^*$.

We may now define the pure/natural direct effect and natural indirect effect \citep{robins1992identifiability, pearl2001direct}, which form the following decomposition of the average causal effect:
\begin{align*}
&E\left\{ Y(a)\right\} -E\left\{ Y(a^*)\right\}  \\
&=\overset{\mathrm{total\text{ }effect}}{\overbrace{E\left[
Y\{a,M(a)\}\right] -E\left[ Y\{a^*,M(a^*)\}\right] }} \\
&=\overset{\mathrm{natural\text{ }indirect\text{ }effect}}{\overbrace{E%
\left[ Y\{a,M(a)\}\right] -E\left[ Y\{a,M(a^*)\}\right] }}+\overset{%
\mathrm{pure\text{ }direct\text{ }effect}}{\overbrace{E\left[
Y\{a,M(a^*)\}\right] -E\left[ Y\{a^*,M(a^*)\}\right) }}.\newline
\end{align*}
The terms $E\{Y(a)\}=E\left[Y\left\{a,M(a)\right\}\right]$, for all $a$, are identified under randomization of $A$. The parameter $\gamma_0\equiv E[Y\{a,M(a^*)\}]$ would be identified if one were to interpret the $\DAG$ in Fig. \ref{fig:1}.(a) as a nonparametric structural equation model with independent errors ($\NPSEM$). Structural equations provide a nonparametric algebraic interpretation of this $\DAG$ corresponding to three equations, one for each variable in the graph. Each random variable on the graph is associated with a distinct, arbitrary function, denoted $g$, and a distinct random disturbance, denoted $\varepsilon$, each with a subscript corresponding to its respective random variable. Each variable is generated by its corresponding function, which depends only on all variables that affect it directly (i.e., its parents on the graph), and its corresponding random disturbance, as follows:
\begin{align*}
A&=g_A(\varepsilon_A)\\
M&=g_M(A,\varepsilon_M)\\
Y&=g_Y(A,M,\varepsilon_Y).
\end{align*}
Under particular interventions, these structural equations naturally encode dependencies of counterfactuals. Consider, for example, two interventions, one setting $A=a^*$, and another setting $A=a$ and $M=m$. The structural equations then become
\begin{equation*}
\begin{aligned}[c]
A&=a^*\\
M(a^*)&=g_M(a^*,\varepsilon_M)\\
Y(a^*)&=g_Y(a^*,M(a^*),\varepsilon_Y)
\end{aligned}
\qquad\qquad\qquad\qquad\qquad
\begin{aligned}[c]
A&=a\\
M(a)&=m\\
Y(a,m)&=g_Y(a,m,\varepsilon_Y).
\end{aligned}
\end{equation*}

This formulation places no a priori restriction on the distribution of counterfactuals. The key assumption of the $\NPSEM$ is that the random disturbances are mutually independent. This allows us to make independence statements regarding counterfactuals under various, possibly-conflicting interventions. In particular, this model implies that for all $m$, (i) $\{M(a),Y(a,m)\}\ci A$, (ii) $Y(a,m)\ci M(a)\mid A=a$, and (iii) $Y(a,m)\ci M(a^*)\mid A=a$, which in turn suffice for identification of $\gamma_0$ \citep{pearl2001direct}. Independence statements such as (iii) are known as cross-world counterfactual statements if $a$ is not equal to $a^*$, due to their comparison of interventions that could never occur in the same world simultaneously. Independence condition (iii) can be seen to hold under the model by considering the $\NPSEM$ under a specific intervention and noting that the only source of randomness in $Y(a,m)=g_Y(a,m,\varepsilon_Y)$ is $\varepsilon_Y$ and the only source of randomness in $M(a^*)=g_M(a^*,\varepsilon_M)$ is $\varepsilon_M$. Thus, the cross-world-counterfactual-independence statement follows directly from independence of exogenous disturbances. However, such an independence is neither experimentally verifiable nor enforceable \citep{robins2010alternative}.

This issue has been discussed extensively \citep{robins2010alternative,richardson2013single}, and in large part motivated the development of the single-world intervention graphs ($\SWIG$s) of \cite{richardson2013single}. These causal graphs manage to elucidate this issue by graphically representing the counterfactuals themselves, allowing independence statements of counterfactuals to be read directly from the graph. Consider the $\SWIG$ in Fig. \ref{fig:1}.(b). By $d$-separation, it is clear that (i) $Y(a,m)\ci M(a)$ for all $a$ and $m$, however no such statement can be made from the graph about $Y(a,m)$ and $M(a^*)$ when $a\neq a^*$. Under this $\SWIG$, independence between $Y(a,m)$ and $M(a^*)$ is not assumed, and hence $\gamma_0$ is not point identified. \cite{robins2010alternative} provide the following bounds for its partial identification in the setting where $M$ is binary and $\SWIG$ independence assumptions $M(a)\ci A$ and $Y(a,m)\ci \{M(a),A\}$ hold for all $a$ and $m$:
\begin{align*}
\max &\{0,\pr(M=0\mid A=a^*)+E(Y\mid M=0,A=a)-1\}\\
+&\max \{0,\pr(M=1\mid A=a^*)+E(Y\mid M=1,A=a)-1\}\\
&\qquad\qquad\qquad\qquad\leq \gamma_0 \leq \\
\min &\{\pr(M=0\mid A=a^*),E(Y\mid M=0,A=a)\}\\
+&\min \{\pr(M=1\mid A=a^*),E(Y\mid M=1,A=a)\}.
\end{align*}
In Section 2, we extend this result to the setting of a polytomous $M$.

As previously mentioned, another often-overlooked condition required for identification of $\gamma_0$ is that there is no confounder of the mediator's effect on the outcome that is affected by the exposure. Such a confounder is present in the setting illustrated in the $\DAG$ in Fig. \ref{fig:2}.(a).
\begin{figure}
\centering
\begin{tabular}{ccc}
\\
\\
\begin{tikzpicture}[->,>=stealth',baseline={(A)},scale=1, line width=1pt]
\tikzstyle{every state}=[draw=none]
\node[shape=circle, draw, inner sep=1mm] (A) at (0,0) {$A$};
\node[shape=circle, draw, inner sep=1mm] (R) at (1.25,0) {$R$};
\node[shape=circle, draw, inner sep=1mm] (M) at (2.5,0) {$M$};
\node[shape=circle, draw, inner sep=1mm] (Y) at (3.75,0) {$Y$};

  \path 
	(A)  edge               (R)  
	(R)  edge               (M)
	(M)  edge               (Y)
	(A)  edge  [bend left]  (M)
	(A)  edge  [bend left=45]  (Y)
	(R)  edge  [bend right]  (Y)

	;
\end{tikzpicture}
& 
&
\begin{tikzpicture}[->,>=stealth',baseline={(A)},scale=1, line width=1pt]
\tikzstyle{every state}=[draw=none]
\node[shape=semicircle, draw, inner sep=1mm, shape border rotate=90, inner sep=1.5mm] (A) at (0,0) {$A$};
\node[shape=semicircle, draw, shape border rotate=270, color=red, inner sep=1.85mm] (a) at (.75,0) {$\tilde{a}$};
\node[shape=ellipse, draw, inner sep=1mm] (R) at (2.75,0) {$R(\tilde{a})$};
\node[shape=semicircle, draw, inner sep=1mm, shape border rotate=90, inner sep=.5mm] (M) at (4.75,0) {$M(\tilde{a})$};
\node[shape=semicircle, draw, shape border rotate=270, color=red, inner sep=2mm] (m) at (5.75,0) {$\tilde{m}$};
\node[shape=ellipse, draw, inner sep=1mm] (Y) at (7.75,0) {$Y(\tilde{a},\tilde{m})$};

  \path 
	(a)  edge               (R)  
	(R)  edge               (M)
	(m)  edge               (Y)  
	(a)  edge  [bend left=60]   (Y)  
	(R)  edge  [bend right=60]   (Y)  
	(a)  edge  [bend left]   (M)  
	;
\end{tikzpicture}

\end{tabular}
\caption{(a) A mediation directed acyclic graph in which $R$ is an exposure-induced confounder. The nodes represent random variables, and the arrows represent possible causal effects of one random variable on another. (b) The single-world intervention graph in the setting of (a) that has been intervened on to set $A$ to $\tilde{a}\in\{a,a^*\}$ and $M$ to $\tilde{m}$. The black nodes represent random variables under this intervention, the red nodes represent the level an intervened random variable takes under this intervention, and the arrows represent possible causal effects of one variable under this intervention on another.}
\label{fig:2}
\end{figure}
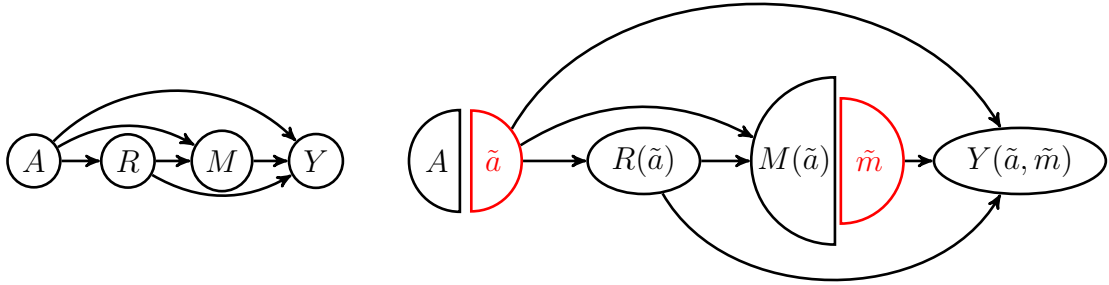
Generally, even under an $\NPSEM$ interpretation of this $\DAG$, $\gamma_0$ will not be identified in this setting. This is readily seen by considering the following representation under this model given by \cite{robins2010alternative}:
\begin{align}
\gamma_0=\sum\limits_{r,r^*}&E\left\{ E(Y\mid M,R=r,A=a)\mid R=r^*,A=a^*\right\}\pr\left\{R(a)=r,R(a^*)=r^*\right\}.
\end{align}
Clearly the joint probability term can never be identified from observed data, since we will never be able to observe $R(a)$ and $R(a^*)$ for the same individual.

A few conditions for identification have been proposed. \cite{robins2010alternative} give two. The first is that $R(a)\ci R(a^*)$, in which case the troublesome term in (1) will factor, giving
\begin{align*}
\gamma_0 = \sum\limits_{r^*,r}&E\left\{ E(Y\mid M,R=r,A=a)\mid R=r^*,A=a^*\right\}\pr(R=r^*\mid A=a^*)\\
&\times\pr(R=r\mid A=a).
\end{align*}
It seems biologically unlikely, however, that in a scenario in which $A$ affects $R$, the counterfactual $R$ under $A=a$ would not be predictive of the counterfactual $R$ under $A=a^*$. The other condition is that the counterfactual outcome under one exposure value is a deterministic function of the counterfactual for the other treatment, i.e., $R(a)=g\{R(a^*)\}$. In this case,
\begin{align*}
\gamma_0 = \sum\limits_{r^*,r}&E\left\{ E(Y\mid M,R=r,A=a)\mid R=r^*,A=a^*\right\}\pr(R=r^*\mid A=a^*)I\{r=g(r^*)\}.
\end{align*}
The above assumption is implied by rank preservation \citep{robins2010alternative}, which is unlikely to hold in social and health sciences as it rules out individual-level effect heterogeneity \citep{tchetgen2014identification}. As none of these conditions are experimentally verifiable, the authors themselves ``do not advocate blithely adopting such assumptions in order to preserve identification of the $\PDE$ in [this setting]" \citep{robins2010alternative}.

\cite{tchetgen2014identification} give two testable conditions for identification of $\gamma_0$ when $R$ is present. The first is of $A$--$R$ monotonicity, i.e., for Bernoulli $R$, $R(a)\geq R(a^*)$. If $R$ is a vector of Bernoulli random variables whose structural equations have independent errors, and if monotonicity holds for each element,
\[\gamma_0=\sum\limits_{r,r^*}E\left\{ E(Y\mid M,R=r,A=a)\mid R=r^*,A=a^*\right\}\prod\limits_{j=1}^k f_j(r_j,r^*_j,a,a^*)\]
where
\begin{align*}
f_j(r_j,r_j^*,a,a^*)=\left\{
\begin{array}{rl}
\pr(R_j=1\mid A=a^*) & \textrm{ if } r_j^*=r_j=1,\\
\pr(R_j=1\mid A=a)-\pr(R_j=1\mid A=a^*) & \textrm{ if } r_j^*=0 \textrm{ and } r_j=1,\\
0 & \textrm{ if } r_j^*=1 \textrm{ and } r_j=0,\\
\pr(R_j=0\mid A=a) & \textrm{ if } r_j^*=r_j=0.
\end{array}\right.
\end{align*}
Their second condition is no $M$--$R$ additive mean interaction, i.e.,
\[E(Y\mid m,r,a)-E(Y\mid m^*,r,a)-E(Y\mid m,r^*,a)+E(Y\mid m^*,r^*,a)=0,\]
for all levels $m$ and $m^*$ of $M$ and $r$ and $r^*$ of $R$. For discrete $M$ and $R$, this yields
\begin{align*}
\gamma_0 = &\sum_m \left\{E(Y\mid m,r^*,a)-E(Y\mid m^*,r^*,a)\right\}\pr(M=m\mid A=a^*)\\
&+\sum_r \left\{E(Y\mid m^*,r,a)-E(Y\mid m^*,r^*,a)\right\}\pr(R=r\mid A=a)\\
&+E(Y\mid m^*,r^*,a).
\end{align*}

Eschewing the cross-world-counterfactual assumptions of the $\NPSEM$ , \cite{tchetgen2014bounds} extend the bounds of \cite{robins2010alternative} to allow for the presence of an exposure-induced confounder when the mediator is binary:
\begin{align*}
\max &\left\{0,\pr(M=0\mid A=a^*)+\sum_r E(Y\mid M=0,R=r,A=a)\pr(R=r\mid A=a)-1\right\}\\
+\max& \left\{0,\pr(M=1\mid A=a^*)+\sum_r E(Y\mid M=1,R=r,A=a)\pr(R=r\mid A=a)-1\right\}\\
&\qquad\qquad\qquad\qquad\qquad\qquad\qquad\leq \gamma_0 \leq \\
\min &\left\{\pr(M=0\mid A=a^*),\sum_r E(Y\mid M=0,R=r,A=a)\pr(R=r\mid A=a)\right\}\\
+\min& \left\{\pr(M=1\mid A=a^*),\sum_r E(Y\mid M=1,R=r,A=a)\pr(R=r\mid A=a)\right\}.
\end{align*}
We extend these bounds as well to allow for polytomous $M$ in Section 3. Additionally, we construct bounds for $\gamma_0$ under an $\NPSEM$ that account for a discrete exposure-induced confounder, but require no further assumption.

%%%%%%%%%%%%%%%%%%%%%%%%%%%%%%%%%%%%%%%%%%%%%%%%%%%%%%%%%%%%%%%%%%%%%%%%%%%%%%%%%%%%%%%%%%%%%%%%%%%%%%%%%
\section{New partial identification results}
%%%%%%%%%%%%%%%%%%%%%%%%%%%%%%%%%%%%%%%%%%%%%%%%%%%%%%%%%%%%%%%%%%%%%%%%%%%%%%%%%%%%%%%%%%%%%%%%%%%%%%%%%
We begin by extending the bounds of \cite{robins2010alternative} and \cite{tchetgen2014bounds} to settings with discrete mediator and outcome. Proofs can be found in the Appendix.

\begin{theorem}
\label{theorem1}
Under the $\SWIG$ in either Fig. \ref{fig:1}.(b) or Fig. \ref{fig:2}.(b) with discrete $M$ and $Y$ and arbitrary $R$,
\begin{align*}
\sum\limits_{m,y}&y\left(\max\left[0,\pr\{M(a^*)=m\}+\pr\{Y(a,m)=y\}-1\right]I(y>0)\right.\\
&\left.+\min\left[\pr\{M(a^*)=m\},\pr\{Y(a,m)=y\}\right]I(y<0)\right)\\
&\qquad\qquad\qquad\qquad\qquad\leq \gamma_0 \leq\\
\sum\limits_{m,y}&y\left(\max\left[0,\pr \{M(a^*)=m\}+\pr\{Y(a,m)=y\}-1\right]I(y<0)\right.\\
&\left.+\min\left[\pr\{M(a^*)=m\},\pr\{Y(a,m)=y\}\right]I(y>0)\right).
\end{align*}
\end{theorem}

The upper and lower bounds coincide when $Y(a,m)$ or $M(a^*)$ is degenerate, which follows from the properties of joint probability mass functions. The upper bound is achieved only if $Y(a,m)$ and $M(a^*)$ are comonotone for each $m$, i.e., if $F_{Y(a,m),M(a^*)}(y,m)=\min\left[F_{Y(a,m)}(y),F_{M(a^*)}(m)\right]$ for each $m$, where $F_X(\cdot)$ denotes the joint (or marginal) cumulative distribution function of the random vector (or scalar) $X$. The lower bound is achieved only if they are countermonotone for each $m$, i.e., if $F_{Y(a,m),M(a^*)}(y,m)=\max\left\{0,F_{Y(a,m)}(y)+F_{M(a^*)}(m)-1\right\}$ for each $m$. A straightforward application of the $g$-formula under the $\DAG$s in Fig. \ref{fig:1} and \ref{fig:2} yields the following corollaries:
\begin{corollary}
For polytomous $M$ and $Y$, $\gamma_0$ is partially identified under the $\SWIG$ in Fig. \ref{fig:1}.(b) by the bounds in Theorem 1 with $\pr\{M(a^*)=m\}=\pr(M=m\mid a^*)$ and $\pr\{Y(a,m)=y\}=\pr(Y=y\mid m,a)$. It is partially identified under the $\SWIG$ in Fig. \ref{fig:2}.(b) by the same bounds, but with $\pr\{M(a^*)=m\}=\pr(M=m\mid a^*)$ and $\pr\{Y(a,m)=y\}=\sum_r\pr(Y=y\mid m,r,a)\pr(R=r\mid a)$.
\end{corollary}
The second part of the corollary continues to hold even if there were a hidden common cause of $R$ and $Y$ as in Fig. \ref{fig:3}, since the same $g$-formula applies in this setting.
\begin{figure}
\centering
\begin{tabular}{ccc}
\\
\\
\begin{tikzpicture}[->,>=stealth',baseline={(A)},scale=1, line width=1pt]
\tikzstyle{every state}=[draw=none]
\node[shape=circle, draw, inner sep=1mm] (A) at (0,0) {$A$};
\node[shape=circle, draw, inner sep=1mm] (R) at (1.25,0) {$R$};
\node[shape=circle, draw, inner sep=1mm] (M) at (2.5,0) {$M$};
\node[shape=circle, draw, inner sep=1mm] (Y) at (3.75,0) {$Y$};
\node[shape=circle, draw, inner sep=1mm, color=gray] (H) at (2.5,1.25) {$H$};

  \path 
	(A)  edge               (R)  
	(R)  edge               (M)
	(M)  edge               (Y)
	(A)  edge  [bend right]  (M)
	(A)  edge  [bend right=45]  (Y)
	(R)  edge  [bend left=35]  (Y)
	(H)  edge	[bend right]               (R)
	(H)  edge	[bend left]               (Y)
	;
\end{tikzpicture}
& 
&
\begin{tikzpicture}[->,>=stealth',baseline={(A)},scale=1, line width=1pt]
\tikzstyle{every state}=[draw=none]
\node[shape=semicircle, draw, inner sep=1mm, shape border rotate=90, inner sep=1.5mm] (A) at (0,0) {$A$};
\node[shape=semicircle, draw, shape border rotate=270, color=red, inner sep=1.85mm] (a) at (.75,0) {$\tilde{a}$};
\node[shape=ellipse, draw, inner sep=1mm] (R) at (2.75,0) {$R(\tilde{a})$};
\node[shape=semicircle, draw, inner sep=1mm, shape border rotate=90, inner sep=.5mm] (M) at (4.75,0) {$M(\tilde{a})$};
\node[shape=semicircle, draw, shape border rotate=270, color=red, inner sep=2mm] (m) at (5.75,0) {$\tilde{m}$};
\node[shape=ellipse, draw, inner sep=1mm] (Y) at (7.75,0) {$Y(\tilde{a},\tilde{m})$};
\node[shape=circle, draw, inner sep=1mm, color=gray] (H) at (5.25,2) {$H$};

  \path 
	(a)  edge               (R)  
	(R)  edge               (M)
	(m)  edge               (Y)  
	(H)  edge	[bend right]               (R)  
	(H)  edge	[bend left]               (Y)  
	(a)  edge  [bend right=60]   (Y)  
	(R)  edge  [bend left=45]   (Y)  
	(a)  edge  [bend right]   (M)  
	;
\end{tikzpicture}
\end{tabular}
\caption{(a) A mediation directed acyclic graph in which an unobserved variable $H$ affects $R$, an exposure-induced confounder, and $Y$. The black nodes represent observed random variables, and the arrows represent possible causal effects of one random variable on another. (b) The single-world intervention graph in the setting of (a) that has been intervened on to set $A$ to $\tilde{a}\in\{a,a^*\}$ and $M$ to $\tilde{m}$. The black nodes represent random variables under this intervention, the red nodes represent the level an intervened random variable takes under this intervention, and the arrows represent possible causal effects of one variable under this intervention on another. In each panel, the gray node represents a hidden random variable}
\label{fig:3}
\end{figure}
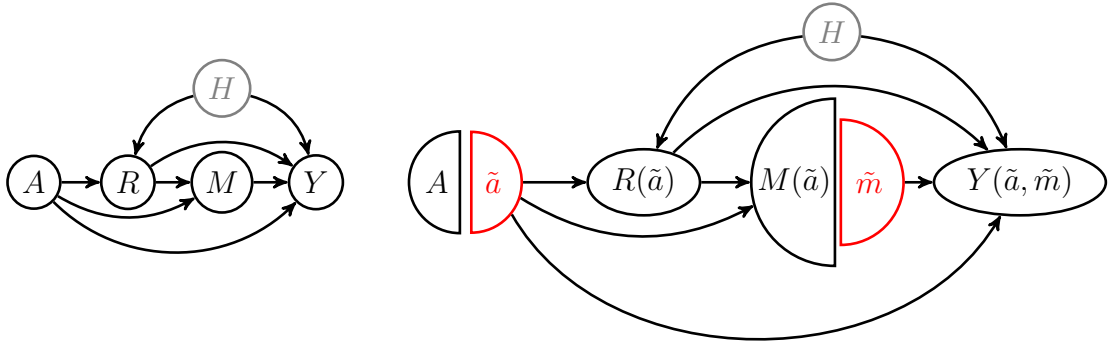
Whereas the previous results invoked no cross-world-counterfactual independences under the $\SWIG$ interpretation of the $\DAG$ in Fig. \ref{fig:2}.(a), sharper bounds are available under Pearl's $\NPSEM$ interpretation of these $\DAG$s, as derived in the following result.

\begin{theorem}
\label{theorem2}
For discrete $R$ taking values in $\{1,\hdots,p\}$, let $B$ be the $p^2\times (p-1)^2$ matrix
\[\left[
\begin{array}{ccccc}
I_{p-1} & 0_{(p-1)\times (p-1)} & \cdots & 0_{(p-1)\times (p-1)} & 0_{(p-1)\times (p-1)} \\
-1_{p-1}^T & 0_{p-1}^T & \cdots & 0_{p-1}^T & 0_{p-1}^T\\
0_{(p-1)\times (p-1)} & I_{p-1} & \cdots & 0_{(p-1)\times (p-1)}  & 0_{(p-1)\times (p-1)}\\
0_{p-1}^T & -1_{p-1}^T & \cdots & 0_{p-1}^T & 0_{p-1}^T\\
\vdots & \vdots & \ddots & \vdots & \vdots \\
0_{(p-1)\times (p-1)} & 0_{(p-1)\times (p-1)} & \cdots & I_{p-1} & 0_{(p-1)\times (p-1)}\\
0_{p-1}^T & 0_{p-1}^T & \cdots & -1_{p-1}^T & 0_{p-1}^T \\
0_{(p-1)\times (p-1)} & 0_{(p-1)\times (p-1)} & \cdots & 0_{(p-1)\times (p-1)} & I_{p-1}\\
0_{p-1}^T & 0_{p-1}^T & \cdots & 0_{p-1}^T & -1_{p-1}^T \\
-I_{p-1} & -I_{p-1} & \cdots & -I_{p-1} & -I_{p-1}\\
&& 1_{(p-1)^2}^T &&
\end{array}\right],\]
$d$ be the $p^2$-dimensional vector
\[
\left[
\begin{array}{c}
0_{p-1}\\
\pr\left(R=1\mid A=a\right)\\
0_{p-1}\\
\pr\left(R=2\mid A=a\right)\\
\vdots\\
0_{p-1}\\
\pr\left(R=p-1\mid A=a\right)\\
\pr\left(R=1\mid A=a^*\right)\\
\pr\left(R=2\mid A=a^*\right)\\
\vdots\\
\pr\left(R=p-1\mid A=a^*\right)\\
\pr\left(R=p\mid A=a\right)+\pr\left(R=p\mid A=a^*\right)-1\\
\end{array}\right],
\]
and $x$ be the vectorization of the matrix $\left[E\left\{ E(Y\mid M,R=r,A=a)\mid R=r^*,A=a^*\right\}\right]_{r,r^*}$. Under a $\NPSEM$ corresponding to the $\DAG$ in Fig. \ref{fig:2}.(a) where $M$ and $Y$ can be either continuous or discrete, $\gamma_0$ is partially identified by $\left[x^T(B\delta_L+d),x^T(B\delta_U+d)\right]$, where $\delta_L$ and $\delta_U$ are the minimizing and maximizing solutions respectively to the linear programming problem with objective function $x^TB\delta$ subject to the constraints
\[\left[
\begin{array}{c}
I_{(p-1)^2} \\
-I_{(p-1)^2}
\end{array}\right]
\delta\leq
\left[
\begin{array}{c}
\min \{\pr(R=1\mid A=a),\pr(R=1\mid A=a^*)\}\\
\min \{\pr(R=1\mid A=a),\pr(R=2\mid A=a^*)\}\\
\vdots\\
\min \{\pr(R=p\mid A=a),\pr(R=p-1\mid A=a^*)\}\\
\min \{\pr(R=p\mid A=a),\pr(R=p\mid A=a^*)\}\\
\min \{0,1-\pr(R=1\mid A=a)-\pr(R=1\mid A=a^*)\}\\
\min \{0,1-\pr(R=1\mid A=a)-\pr(R=2\mid A=a^*)\}\\
\vdots\\
\min \{0,1-\pr(R=p\mid A=a)-\pr(R=p-1\mid A=a^*)\}\\
\min \{0,1-\pr(R=p\mid A=a)-\pr(R=p\mid A=a^*)\}
\end{array}\right]\]
and $\delta\geq 0$. 
\end{theorem}

Similar to the previous result, these bounds coincide if either $R(a)$ or $R(a^*)$ is degenerate. The upper bound is achieved when $R(a)$ and $R(a^*)$ are comonotone; the lower bound is achieved when they are countermonotone. While these bounds are not available in closed form, they can be readily solved using standard software, such as with the lp\_solve function, which uses the revised simplex method and is accessible from a number of languages, including R, MATLAB, Python, and C. While the method used by this software is not guaranteed to converge at a polynomial rate \citep{klee1970good}, it is quite efficient in most cases \citep{schrijver1998theory}. The following corollary shows that these bounds reduce to a closed form when $R$ is binary.

\begin{corollary}
Under a $\NPSEM$ corresponding to the $\DAG$ in Fig. \ref{fig:2}.(a) with binary $R$,
\begin{align*}
\min\limits_{\pi_{11}\in\Pi}\sum\limits_{r,r^*}&E\left\{ E(Y\mid M,R=r,A=a)\mid R=r^*,A=a^*\right\}h(r,r^*,\pi_{11})\\
&\qquad\qquad\qquad\qquad\leq\gamma_0\leq\\
\max\limits_{\pi_{11}\in\Pi}\sum\limits_{r,r^*}&E\left\{ E(Y\mid M,R=r,A=a)\mid R=r^*,A=a^*\right\}h(r,r^*,\pi_{11})\\
\end{align*}
where $\Pi$ is the set
\begin{align*}
\{&\max\left\{0,\pr(R=1\mid A=a)+\pr(R=1\mid A=a^*)-1\right\},\\
&\min\left\{\pr(R=1\mid A=a),\pr(R=1\mid A=a^*)\right\}\}
\end{align*}
and
\begin{align*}
h(r,r^*,\pi_{11})=\left\{
\begin{array}{rl}
\pi_{11} & \textrm{ if } r^*=r=1,\\
\pr(R=1\mid A=a)-\pi_{11} & \textrm{ if } r^*=0 \textrm{ and } r=1,\\
\pr(R=1\mid A=a^*)-\pi_{11} & \textrm{ if } r^*=1 \textrm{ and } r=0,\\
1-\pr(R=1\mid A=a)-\pr(R=1\mid A=a^*)+\pi_{11} & \textrm{ if } r^*=r=0.
\end{array}\right.
\end{align*}
\end{corollary}

Under $A-R$ monotonicity with binary $R$, the identifying functional given by \cite{tchetgen2014identification} is recovered at the upper bound in Corollary 2. All results given here can be extended to settings with observed pre-exposure confounders, which we denote $C$. In Corollary 1, one must first perform conditional inference given C, then subsequently average over the conditional bounds. This is in fact valid due to Jensen's inequality, because the constraints on the marginal joint probabilities are already implied by the constraints enforced on the conditional joint distributions, so no further constraints need be considered. However, Jensen's inequality does not apply in the case of Theorem 2, so controlling for $C$ requires estimating two pairs of candidate bounds and selecting the larger of the lower bounds and the smaller of the upper bounds. When $p$ is of moderate size, $\delta$ can be solved for each covariate pattern of $C$, i.e., without modeling the dependence of the cross-world-counterfactual joint distribution on $C$. Averaging the resulting conditional bounds gives the first pair of bounds. The second pair results from replacing each probability in the theorem with an average over the probabilities conditional on $C$ and doing the same with $x$.

%%%%%%%%%%%%%%%%%%%%%%%%%%%%%%%%%%%%%%%%%%%%%%%%%%%%%%%%%%%%%%%%%%%%%%%%%%%%%%%%%%%%%%%%%%%%%%%%%%%%%%%%%
\section{Application to Harvard PEPFAR data set}
%%%%%%%%%%%%%%%%%%%%%%%%%%%%%%%%%%%%%%%%%%%%%%%%%%%%%%%%%%%%%%%%%%%%%%%%%%%%%%%%%%%%%%%%%%%%%%%%%%%%%%%%%
We now consider an application to a data set collected by the Harvard President's Emergency Plan for AIDS Relief (PEPFAR) program in Nigeria. The data set consists of previously antiretroviral therapy (ART)-na{\"i}ve, HIV-1 infected adult patients who began ART in the program and were followed at least one year following initiation. Patients without reliable viral load data at two of the hospitals were excluded. Only complete cases initially prescribed to either TDF+3TC/FTC+NVP or AZT+3TC+NVP\footnote{3TC=lamivudine, AZT=zidovudine, FTC=emtricitabine, NVP=nevirapine, TDF=tenofovir} were considered for this analysis. Thus, the data set we consider consists of 6627 patients, 1919 of whom were prescribed to TDF+3TC/FTC+NVP, and the remaining 4708 prescribed to AZT+3TC+NVP.

There has accumulated evidence of a differential effect on virologic failure between these two first-line antiretroviral treatment regimens \citep{tang2012review}. Virologic failure is defined by the World Health Organization as repeat viral load above 1000 copies/mL. We base this on measurements at 12 and 18 months of ART duration in our analysis.

A natural question of scientific interest is what role adherence plays in mediating this differential effect. We are primarily interested in learning about the scientific mechanism of this effect on the individual level. The natural indirect effect best captures this mechanism in that it captures an isolated effect difference mediated by adherence by, in a sense, deactivating effect differences along all other possible causal pathways. We specifically examine the effect through adherence over the second six months since treatment assignment, i.e., the six months prior to the first viral load measurement. Identification is complicated by the presence of treatment toxicity, which clearly affects adherence directly, and has the potential to modify the effect of the treatment assignment on virologic failure. Thus, toxicity measured at six months after treatment assignment is an exposure-induced confounder of the effect of the mediator on the outcome. Further, toxicity and virologic failure are likely to be rendered dependent by unobserved underlying biological common causes as in Fig. \ref{fig:3}, where $H$ represents these hidden biological mechanisms. Because we define the mediator to be adherence over the second six months, adherence over the first six months is also an exposure-induced confounder along with toxicity, and must be accounted for. Had we defined the mediator to be adherence over the full year, measurement of the mediator and toxicity would have overlapped, violating the principle of temporal ordering.

Let $C$ denote the vector consisting of baseline covariates sex, age, marital status, WHO stage, hepatitis C virus, hepatitis B virus, CD4+ cell count, and viral load. Let $A$ be an indicator of ART assignment taking levels $a^*$ for TDF+3TC/FTC+NVP and $a$ for AZT+3TC+NVP; $R$ be a vector of two indicator variables, one of the presence of any lab toxicity, and one of adherence exceeding 95\%, both over the first six months following initiation of therapy; $M$ be an indicator of adherence exceeding 95\% over the subsequent six months; and $Y$ be an indicator of virologic failure at one year, i.e., repeat viral load above 1000 copies/mL at one year and at 18 months.

Here we estimate the natural indirect effect of $A$ on $Y$ through $M$, as defined above, on the risk difference scale using the various sets of identifying assumptions given above. Throughout, inference is performed using maximum likelihood for point estimation and a weighted bootstrap \citep{rao1992approximation,van1996weak} for confidence intervals, which appropriately accounts for the rare outcome. The results are summarized in Fig. \ref{fig:4}.
\begin{figure}%[h]
\centering
\includegraphics[scale=.55]{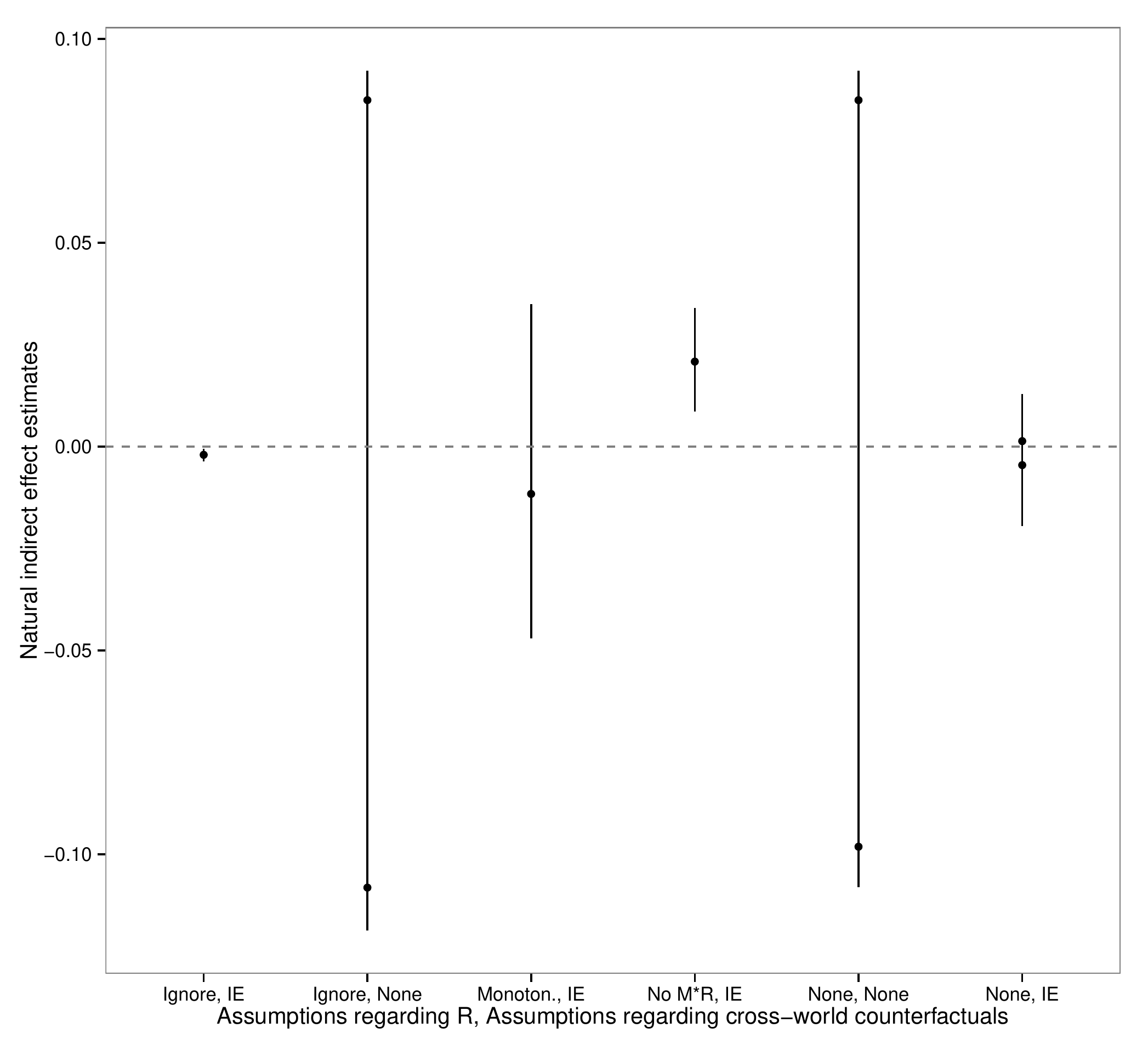}
\caption{A plot showing the estimated natural indirect effect of ART assignment on virologic failure with respect to adherence under the various assumptions. The assumptions vary across the horizontal axis, with the first part of the label indicating the assumption regarding the exposure-induced confounder, $R$, and the second part indicating the assumption regarding cross-world counterfactuals. For the assumptions regarding $R$, ``Ignore" means that the presence of $R$ is ignored altogether, ``Monoton." means the $A$--$R$ monotonicity assumption in Section 1, ``No M*R" means the no $M$--$R$ interaction assumption in Section 1, and ``None" means that $R$ was accounted for without additional assumptions. For the assumptions regarding cross-world counterfactuals, ``IE" means a $\NPSEM$ was assumed, and ``None" means no cross-world-counterfactuals independences were assumed. When the assumptions give partial identification, the two dots represent the point estimates of the upper and lower bound for the natural indirect effect, and the vertical bar represents the bootstrap 95\% confidence interval for the interval. When the assumptions give full identification, the single dot represents the point estimate of the natural indirect effect, and the vertical bar represents its bootstrap 95\% confidence interval.}
\label{fig:4}
\end{figure}
It is immediately apparent that inference is sensitive to which identifying assumptions are made. Consider an investigator who might be willing to rely on cross-world-counterfactual independences. If she decides to ignore the presence of toxicity, she might likely conclude that there is a very small, yet significant negative indirect effect. Conversely, were she to make the no $M$--$R$ interaction assumption, she would find a significant positive indirect effect with considerable uncertainty. In fact, an empirical test of this assumption reveals that it is unlikely to apply. Likewise, the data suggest that the required assumption of independent errors of the components of $R$ is also unlikely to hold. Nonetheless, we present both results for the sake of comparison. Results are fairly imprecise under monotonicity, and do not show a significant effect.

Another investigator unwilling to impose cross-world-counterfactual-independence assumptions is left with little to say as the bounds are wide, and include the null hypothesis of no $\NIE$, regardless of how toxicity is handled. Interestingly, the bounds that result from making no assumptions about the joint distribution of the cross-world $R$ counterfactuals are narrower than the bounds that result from ignoring $R$. That is, the bounds themselves appear narrower; the variances of the interval estimates appear to be comparable. This is because even though we do not impose any restrictions on the distribution of $R$ or its counterfactuals a priori, observing $R$ is clearly informative. The bounds accounting for $R$ have the added advantage of being the only identifying formula that remains valid when toxicity and virologic suppression are affected by an unobserved common cause, as in Fig. \ref{fig:3}.

Finally, incorporating $R$ results in narrower interval estimates than not imposing the $\NPSEM$, even if $R$ were ignored. Thus, cross-world-counterfactual-independences appear to have stronger empirical implications in the current analysis than assumptions regarding exposure-induced confounders. The general trend in these results is that little is gained in terms of precision by assumptions regarding $R$. In fact, the confidence interval for the bounds resulting from the independent errors assumption and no assumption regarding $R$ is narrower than the confidence interval for the estimate that results from assuming monotonicity, despite the fact that the $\NIE$ is point-identified in the latter case. The na\"{\i}ve assumption that $R$ is not a confounder is the only assumption about $R$ under which precision is gained.

\appendix

\section*{Appendix}
\subsection*{Proofs of theorems}

\begin{proof}[Proof of Theorem~\ref{theorem1}]
For each level $m$ and $y$, define $\pi_1(m,y)=\pr\{Y(a,m)=y\}$ and $\pi_2(m)=\pr\{M(a^*)=m\}$. There exist $U_1(m,y), U_2(m)\sim\mathcal{U}(0,1)$ such that $I\{Y(a,m)=y\}=I\{U_1(m,y)\leq\pi_1(m,y)\}$ and $I\{M(a^*)=m\}=I\{U_2(m)\leq\pi_2(m)\}$. The joint distribution $F_{U_1(m,y),U_2(m)}$, then, is a bivariate copula, for which Fr\'{e}chet--Hoeffding sharp bounds exist. Applying these to $\pr\left\{Y(a,m)=y,M(a^*)=m\right\}=F_{U_1(m,y),U_2(m)}\left\{\pi_1(m,y),\pi_2(m)\right\}$, we have
\begin{align*}
\max &\left[0,\pr\{M(a^*)=m\}+\pr\{Y(a,m)=y\}-1\right]\\
&\leq\pr\left\{Y(a,m)=y,M(a^*)=m\right\}\leq\\
\min &\left[\pr\{M(a^*)=m\},\pr\{Y(a,m)=y\}\right].
\end{align*}
Applying these bounds to each summand in
\[E[Y\{a,M(a^*)\}]=\sum\limits_{m,y}y\pr\{Y(a,m)=y,M(a^*)=m\}\]
yields the result.
\end{proof}

\begin{proof}[Proof of Theorem~\ref{theorem2}]
Let $\pi_{r,r^*} = \pr\left\{R(a)=r,R(a^*)=r^*\right\}$, $\pi$ be the vectorization of the matrix $[\pi_{r,r^*}]$, and $\delta$ be the vectorization of the matrix $[\pi_{r,r^*}]_{-p,-p}$, i.e., the vectorization of the matrix $\pi$ with row $p$ and column $p$ removed. Equation (1) can now be expressed as $\gamma_0=x^T\pi$, which is identified in $x$, but not $\pi$. Conditional on the marginal probabilities, which are identified, the joint probabilities have $(p-1)^2$ degrees of freedom, and can be expressed as $\pi=B\delta+d$. Since $x^TB\delta$ is linear in $\delta$ and each element of $\delta$ is constrained by
\begin{align*}
\max&\left\{0,\pr(R=r\mid A=a)+\pr(R=r^*\mid A=a^*)-1\right\}\\
&\qquad\qquad\qquad\leq\pi_{r,r^*}\leq\\
\min&\left\{\pr(R=r\mid A=a),\pr(R=r^*\mid A=a^*)\right\},
\end{align*}
the proposed linear programming problem will yield the $\delta$ that optimizes $x^TB\delta$, and hence $x^T(B\delta+d)$. Thus, $\gamma_0$ will be bounded by $x^T(B\delta+d)$ evaluated at the minimizing and maximizing linear programming solutions $\delta_L$ and $\delta_U$.
\end{proof}

\newpage

\bibliographystyle{apalike}

\bibliography{references}

\begin{thebibliography}{}

\bibitem[Albert, 2008]{albert2008mediation}
Albert, J.~M. (2008).
\newblock Mediation analysis via potential outcomes models.
\newblock {\em Statistics in Medicine}, 27(8):1282--1304.

\bibitem[Albert, 2012]{albert2012mediation}
Albert, J.~M. (2012).
\newblock Mediation analysis for nonlinear models with confounding.
\newblock {\em Epidemiology (Cambridge, Mass.)}, 23(6):879.

\bibitem[Albert and Nelson, 2011]{albert2011generalized}
Albert, J.~M. and Nelson, S. (2011).
\newblock Generalized causal mediation analysis.
\newblock {\em Biometrics}, 67(3):1028--1038.

\bibitem[Albert and Wang, 2015]{albert2015sensitivity}
Albert, J.~M. and Wang, W. (2015).
\newblock Sensitivity analyses for parametric causal mediation effect
  estimation.
\newblock {\em Biostatistics}, 16(2):339--351.

\bibitem[Avin et~al., 2005]{avin2005identifiability}
Avin, C., Shpitser, I., and Pearl, J. (2005).
\newblock Identifiability of path-specific effects.
\newblock In {\em IJCAI-05, Proceedings of the Nineteenth International Joint
  Conference on Artificial Intelligence}, pages 357--363.

\bibitem[Balke and Pearl, 1997]{balke1997probabilistic}
Balke, A.~A. and Pearl, J. (1997).
\newblock Probabilistic counterfactuals: semantics, computation, and
  applications.
\newblock Technical report, DTIC Document.

\bibitem[Cai et~al., 2008]{cai2008bounds}
Cai, Z., Kuroki, M., Pearl, J., and Tian, J. (2008).
\newblock Bounds on direct effects in the presence of confounded intermediate
  variables.
\newblock {\em Biometrics}, 64(3):695--701.

\bibitem[Cheng and Small, 2006]{cheng2006bounds}
Cheng, J. and Small, D.~S. (2006).
\newblock Bounds on causal effects in three-arm trials with non-compliance.
\newblock {\em Journal of the Royal Statistical Society: Series B (Statistical
  Methodology)}, 68(5):815--836.

\bibitem[Goetgeluk et~al., 2008]{goetgeluk2008estimation}
Goetgeluk, S., Vansteelandt, S., and Goetghebeur, E. (2008).
\newblock Estimation of controlled direct effects.
\newblock {\em Journal of the Royal Statistical Society: Series B (Statistical
  Methodology)}, 70(5):1049--1066.

\bibitem[Hsu et~al., 2015]{hsu2015surrogate}
Hsu, J.~Y., Kennedy, E.~H., Roy, J.~A., Stephens-Shields, A.~J., Small, D.~S.,
  and Joffe, M.~M. (2015).
\newblock Surrogate markers for time-varying treatments and outcomes.
\newblock {\em Clinical Trials}, page 1740774515583500.

\bibitem[Imai et~al., 2010a]{imai2010general}
Imai, K., Keele, L., and Tingley, D. (2010a).
\newblock A general approach to causal mediation analysis.
\newblock {\em Psychological Methods}, 15(4):309.

\bibitem[Imai et~al., 2010b]{imai2010identification}
Imai, K., Keele, L., and Yamamoto, T. (2010b).
\newblock Identification, inference and sensitivity analysis for causal
  mediation effects.
\newblock {\em Statistical Science}, pages 51--71.

\bibitem[Kaufman et~al., 2005]{kaufman2005improved}
Kaufman, S., Kaufman, J.~S., MacLehose, R.~F., Greenland, S., and Poole, C.
  (2005).
\newblock Improved estimation of controlled direct effects in the presence of
  unmeasured confounding of intermediate variables.
\newblock {\em Statistics in Medicine}, 24(11):1683--1702.

\bibitem[Klee and Minty, 1970]{klee1970good}
Klee, V. and Minty, G.~J. (1970).
\newblock How good is the simplex algorithm.
\newblock Technical report, DTIC Document.

\bibitem[Pearl, 2001]{pearl2001direct}
Pearl, J. (2001).
\newblock Direct and indirect effects.
\newblock In {\em Proceedings of the Seventeenth Conference on Uncertainty in
  Artificial Intelligence}, pages 411--420. Morgan Kaufmann Publishers Inc.

\bibitem[Petersen et~al., 2006]{petersen2006estimation}
Petersen, M.~L., Sinisi, S.~E., and van~der Laan, M.~J. (2006).
\newblock Estimation of direct causal effects.
\newblock {\em Epidemiology}, 17(3):276--284.

\bibitem[Rao and Zhao, 1992]{rao1992approximation}
Rao, C.~R. and Zhao, L. (1992).
\newblock Approximation to the distribution of {M}-estimates in linear models
  by randomly weighted bootstrap.
\newblock {\em Sankhy{\=a}: The Indian Journal of Statistics, Series A}, pages
  323--331.

\bibitem[Richardson and Robins, 2013]{richardson2013single}
Richardson, T.~S. and Robins, J.~M. (2013).
\newblock Single world intervention graphs ({SWIG}s): A unification of the
  counterfactual and graphical approaches to causality.
\newblock {\em Center for the Statistics and the Social Sciences, University of
  Washington Series. Working Paper}, (128).

\bibitem[Robins, 1989]{robins1989analysis}
Robins, J.~M. (1989).
\newblock The analysis of randomized and non-randomized {AIDS} treatment trials
  using a new approach to causal inference in longitudinal studies.
\newblock {\em Health service research methodology: a focus on AIDS}, 113:159.

\bibitem[Robins, 1999]{robins1999testing}
Robins, J.~M. (1999).
\newblock Testing and estimation of direct effects by reparameterizing directed
  acyclic graphs with structural nested models.
\newblock {\em Computation, Causation, and Discovery}, pages 349--405.

\bibitem[Robins, 2003]{robins2003semantics}
Robins, J.~M. (2003).
\newblock Semantics of causal {DAG} models and the identification of direct and
  indirect effects.
\newblock {\em Highly Structured Stochastic Systems}, pages 70--81.

\bibitem[Robins and Greenland, 1992]{robins1992identifiability}
Robins, J.~M. and Greenland, S. (1992).
\newblock Identifiability and exchangeability for direct and indirect effects.
\newblock {\em Epidemiology}, pages 143--155.

\bibitem[Robins and Richardson, 2010]{robins2010alternative}
Robins, J.~M. and Richardson, T.~S. (2010).
\newblock Alternative graphical causal models and the identification of direct
  effects.
\newblock {\em Causality and Psychopathology: Finding the Determinants of
  Disorders and Their Cures}, pages 103--158.

\bibitem[Rubin, 1974]{rubin1974estimating}
Rubin, D.~B. (1974).
\newblock Estimating causal effects of treatments in randomized and
  nonrandomized studies.
\newblock {\em Journal of Educational Psychology}, 66(5):688.

\bibitem[Rubin, 1978]{rubin1978bayesian}
Rubin, D.~B. (1978).
\newblock Bayesian inference for causal effects: The role of randomization.
\newblock {\em The Annals of Statistics}, pages 34--58.

\bibitem[Schrijver, 1998]{schrijver1998theory}
Schrijver, A. (1998).
\newblock {\em Theory of linear and integer programming}.
\newblock John Wiley \& Sons.

\bibitem[Shpitser, 2013]{shpitser2013counterfactual}
Shpitser, I. (2013).
\newblock Counterfactual graphical models for longitudinal mediation analysis
  with unobserved confounding.
\newblock {\em Cognitive Science}, 37(6):1011--1035.

\bibitem[Sj{\"o}lander, 2009]{sjolander2009bounds}
Sj{\"o}lander, A. (2009).
\newblock Bounds on natural direct effects in the presence of confounded
  intermediate variables.
\newblock {\em Statistics in Medicine}, 28(4):558--571.

\bibitem[Splawa-Neyman et~al., 1990]{splawa1990application}
Splawa-Neyman, J., Dabrowska, D., Speed, T., et~al. (1990).
\newblock On the application of probability theory to agricultural experiments.
  {E}ssay on principles. {S}ection 9.
\newblock {\em Statistical Science}, 5(4):465--472.

\bibitem[Taguri and Chiba, 2015]{taguri2015principal}
Taguri, M. and Chiba, Y. (2015).
\newblock A principal stratification approach for evaluating natural direct and
  indirect effects in the presence of treatment-induced intermediate
  confounding.
\newblock {\em Statistics in Medicine}, 34(1):131--144.

\bibitem[Tang et~al., 2012]{tang2012review}
Tang, M.~W., Kanki, P.~J., and Shafer, R.~W. (2012).
\newblock A review of the virological efficacy of the 4 {W}orld {H}ealth
  {O}rganization--recommended tenofovir-containing regimens for initial {HIV}
  therapy.
\newblock {\em Clinical Infectious Diseases}, 54(6):862--875.

\bibitem[Taylor et~al., 2005]{taylor2005counterfactual}
Taylor, J.~M., Wang, Y., and Thi{\'e}baut, R. (2005).
\newblock Counterfactual links to the proportion of treatment effect explained
  by a surrogate marker.
\newblock {\em Biometrics}, 61(4):1102--1111.

\bibitem[Tchetgen~Tchetgen, 2011]{tchetgen2011causal}
Tchetgen~Tchetgen, E.~J. (2011).
\newblock On causal mediation analysis with a survival outcome.
\newblock {\em The International Journal of Biostatistics}, 7(1):1--38.

\bibitem[Tchetgen~Tchetgen, 2013]{tchetgen2013inverse}
Tchetgen~Tchetgen, E.~J. (2013).
\newblock Inverse odds ratio-weighted estimation for causal mediation analysis.
\newblock {\em Statistics in Medicine}, 32(26):4567--4580.

\bibitem[Tchetgen~Tchetgen and Phiri, 2014]{tchetgen2014bounds}
Tchetgen~Tchetgen, E.~J. and Phiri, K. (2014).
\newblock Bounds for pure direct effect.
\newblock {\em Epidemiology}, 25(5):775--776.

\bibitem[Tchetgen~Tchetgen and Shpitser, 2012]{tchetgen2012semiparametric}
Tchetgen~Tchetgen, E.~J. and Shpitser, I. (2012).
\newblock Semiparametric theory for causal mediation analysis: Efficiency
  bounds, multiple robustness and sensitivity analysis.
\newblock {\em The Annals of Statistics}, 40(3):1816--1845.

\bibitem[Tchetgen~Tchetgen and Shpitser, 2014]{tchetgen2014estimation}
Tchetgen~Tchetgen, E.~J. and Shpitser, I. (2014).
\newblock Estimation of a semiparametric natural direct effect model
  incorporating baseline covariates.
\newblock {\em Biometrika}, 101(4):849--864.

\bibitem[Tchetgen~Tchetgen and VanderWeele, 2014]{tchetgen2014identification}
Tchetgen~Tchetgen, E.~J. and VanderWeele, T.~J. (2014).
\newblock On identification of natural direct effects when a confounder of the
  mediator is directly affected by exposure.
\newblock {\em Epidemiology (Cambridge, Mass.)}, 25(2):282.

\bibitem[Ten~Have et~al., 2007]{ten2007causal}
Ten~Have, T.~R., Joffe, M.~M., Lynch, K.~G., Brown, G.~K., Maisto, S.~A., and
  Beck, A.~T. (2007).
\newblock Causal mediation analyses with rank preserving models.
\newblock {\em Biometrics}, 63(3):926--934.

\bibitem[van~der Laan and Petersen, 2008]{van2008direct}
van~der Laan, M.~J. and Petersen, M.~L. (2008).
\newblock Direct effect models.
\newblock {\em The International Journal of Biostatistics}, 4(1):1--27.

\bibitem[van~der Vaart and Wellner, 1996]{van1996weak}
van~der Vaart, A.~W. and Wellner, J.~A. (1996).
\newblock {\em Weak Convergence and Empirical Processes}.
\newblock Springer.

\bibitem[VanderWeele, 2009]{vanderweele2009marginal}
VanderWeele, T.~J. (2009).
\newblock Marginal structural models for the estimation of direct and indirect
  effects.
\newblock {\em Epidemiology}, 20(1):18--26.

\bibitem[VanderWeele, 2011]{vanderweele2011causal}
VanderWeele, T.~J. (2011).
\newblock Causal mediation analysis with survival data.
\newblock {\em Epidemiology (Cambridge, Mass.)}, 22(4):582.

\bibitem[VanderWeele and Vansteelandt, 2009]{vanderweele2009conceptual}
VanderWeele, T.~J. and Vansteelandt, S. (2009).
\newblock Conceptual issues concerning mediation, interventions and
  composition.
\newblock {\em Statistics and its Interface}, 2:457--468.

\bibitem[VanderWeele and Vansteelandt, 2010]{vanderweele2010odds}
VanderWeele, T.~J. and Vansteelandt, S. (2010).
\newblock Odds ratios for mediation analysis for a dichotomous outcome.
\newblock {\em American Journal of Epidemiology}, 172(12):1339--1348.

\bibitem[Vansteelandt and VanderWeele, 2012]{vansteelandt2012natural}
Vansteelandt, S. and VanderWeele, T.~J. (2012).
\newblock Natural direct and indirect effects on the exposed: effect
  decomposition under weaker assumptions.
\newblock {\em Biometrics}, 68(4):1019--1027.

\bibitem[Wang and Albert, 2012]{wang2012estimation}
Wang, W. and Albert, J.~M. (2012).
\newblock Estimation of mediation effects for zero-inflated regression models.
\newblock {\em Statistics in Medicine}, 31(26):3118--3132.

\bibitem[Wang et~al., 2013]{wang2013estimation}
Wang, W., Nelson, S., and Albert, J.~M. (2013).
\newblock Estimation of causal mediation effects for a dichotomous outcome in
  multiple-mediator models using the mediation formula.
\newblock {\em Statistics in Medicine}, 32(24):4211--4228.

\bibitem[Zhang and Rubin, 2003]{zhang2003estimation}
Zhang, J.~L. and Rubin, D.~B. (2003).
\newblock Estimation of causal effects via principal stratification when some
  outcomes are truncated by “death”.
\newblock {\em Journal of Educational and Behavioral Statistics},
  28(4):353--368.

\end{thebibliography}

\end{document}